\documentclass[10pt,twocolumn,twoside,journal,letterpaper]{IEEEtran}%
\usepackage[square,comma,sort&compress,numbers]{natbib}

\usepackage{amsmath}
\usepackage{amsfonts}
\usepackage{amssymb}
\usepackage{amsthm}
\usepackage{mathtools}
\mathtoolsset{showonlyrefs}
\usepackage[detect-weight=true]{siunitx}
\usepackage{nicefrac}
\usepackage{graphicx}
\usepackage{empheq}
\usepackage{comment}
\usepackage{xcolor}
\usepackage[hidelinks]{hyperref}%

\newtheoremstyle{problemStyle}{\topsep}{\topsep}{\itshape}{0pt}{\bfseries}{:}{5pt plus 1pt minus 1pt}{\thmname{#1}\thmnumber{ #2} \thmnote{(#3)}}
\newtheoremstyle{noDotsStyle}{\topsep}{\topsep}{\itshape}{0pt}{\bfseries}{}{5pt plus 1pt minus 1pt}{\thmname{#1}\thmnumber{ #2} \thmnote{(#3)}}

\theoremstyle{problemStyle}
\newtheorem{problem}{Problem}
\newtheorem{lemma}{Lemma}

\newtheorem{proposition}{Proposition}

\theoremstyle{noDotsStyle}
\newtheorem{problemNoDots}[problem]{Problem}

\begin{document}

\title{Resource Allocation in MIMO Radar With Multiple Targets for Non-Coherent Localization}
\author{
	Nil~Garcia,
	Alexander~M.~Haimovich,~\IEEEmembership{Fellow,~IEEE},
	Martial~Coulon,
	and Marco~Lops,~\IEEEmembership{Senior~Member,~IEEE}%
	\thanks{
		This work was supported in part by the U.S.\ Air Force Office of Scientific Research under agreement FA9550-09-1-0303.
	}%
	\thanks{
		N.~Garcia is with the Electrical and Computer Engineering Department, New Jersey Institute of Technology, Newark, NJ 07102 USA, and the T\'{e}SA laboratory, 31000 Toulouse, France (e-mail: nil.garcia@njit.edu).
	}%
	\thanks{
		A.M.~Haimovich is with the Electrical and Computer Engineering Department, New Jersey Institute of Technology, Newark, NJ 07102 USA (e-mail: haimovich@njit.edu).
	}%
	\thanks{
		M.~Coulon is with the IRIT/INP-ENSEEIHT, University of Toulouse, BP 7122, 31071 Toulouse Cedex 7, France (e-mail: martial.coulon@enseeiht.fr).
	}%
	\thanks{
		M.~Lops is with the DAEIMI, Università degli Studi di Cassino, 03043, Italy (e-mail: e.grossi@unicas.it; lops@unicas.it).
	}
}

\maketitle

\begin{abstract}

In a MIMO radar network the multiple transmit elements may emit waveforms that differ on power and bandwidth. In this paper, we are asking, given that these two resources are limited, what is the optimal power, optimal bandwidth and optimal joint power and bandwidth allocation for best localization of multiple targets. The well known Cr\'amer-Rao lower bound for target localization accuracy is used as a figure of merit and approximate solutions are found by minimizing a sequence of convex problems. Their quality is assessed through extensive numerical simulations and with the help of a lower-bound on the true solution. Simulations results reveal that bandwidth allocation policies have a definitely stronger impact on performance than power.

\end{abstract}


\section{Introduction}

In a MIMO radar architecture, transmit elements emit their respective signals, which are subsequently scattered by targets in the field of view towards receive elements. The elements at the receive and the transmit ends may be colocated, so that the transmit-receive paths generated by each target share the same phase and amplitude \cite{Li06,Xu06,Li07}. Conversely, if the antenna elements are widely spaced, then they view the targets from different aspect angles, making these (substantially diverse) paths exhibit different amplitudes \cite{Haimovich08,Fishler04-2,Fishler04,Fishler06}. In either case, the signals emitted by the transmitters can be chosen independently, thus adding degrees of freedom \cite{Bliss03} as a precious resource for such tasks as targets localization and tracking.
Due to the lack in angle diversity in the case of colocated antennas, a point target model is suitable, thus enabling coherent processing that exploits the phase coherence between the paths \cite {Li07,Robey05,Lehmann06}. In contrast, in the widely spaced elements case, the point target model breaks down in favor of an extended target model \cite{Mitchell74}, because targets display different radar cross-sections (RCS) in different directions. In this case, processing is non-coherent and angular diversity becomes an object of interest \cite{Haimovich08,He10,He10-2}.

In a non-coherent MIMO system where phase is not preserved across the elements, and if the system elements are time-synchronized, targets may be localized by multilateration \cite{Godrich09,Wang11} or more computation-intensive techniques such as Direct Position Determination (PDP) \cite{Bar-Shalom11}.
When the target reflections towards different directions are unknown, performance is enhanced by illuminating the target from different angles, and averaging the target scintillations \cite{Fishler06}. In this work, we are concerned with improving the accuracy of the localization of a detected, stationary target. The premise of the work is that, judicious allocation of system resources, power and bandwidth, has an impact on the accuracy of localization. We formulate cost functions, develop algorithms and analyze performance of resource allocation methods addressing the question of optimal resource allocation for target localization with non-coherent MIMO radar.

Besides high-accuracy localization, 
resource-aware design is of importance in surveillance radars with mounted mobiles stations powered off-grid, as it can be found in network-centric warfare \cite{Dekker08}. Such a configuration with multiple transmitters and receivers is robust to the loss of nodes, e.g.\ due to hostile action. Furthermore, resource management is an essential part of military operations in hostile environments, where low-probability-of-intercept operation may be required \cite{Nelms11}.
Another growing field where non-coherent MIMO might be applied, is ultra wide band (UWB) radar sensor networks \cite{Gezici05}, which usually operate under severe power constraints because they reuse wireless communications equipment or operate in unregulated frequencies. The ability of UWB to perform through-the-wall detection makes it an attractive tool for detecting intruders in buildings \cite{Paolini08}. Another example of non-coherent radar localization, very similar to our system as it will become clear in the next paragraphs, is presented in \cite{Tong12}, where an air traffic control radar uses multiple transmitters operating at different frequencies, and multilateration techniques are used for tracking.

A criterion for measuring the target location accuracy is necessary for allocating the resources. As a cost function, we choose the Cr\'amer-Rao lower bound (CRLB) for the estimation of a target location in a distributed architecture derived in \cite{Godrich10}. An advantage of this cost function is that it is in closed form, thus making it suitable for algebraic manipulations. Also, the CRLB is known to provide a tight lower bound on the error of an unbiased estimator at high signal-to-noise ratio (SNR) \cite{He10-2,Kay}, and this is the SNR regime in which we operate. And lastly, another important particularity of this CRLB is that, under the assumption of full orthogonality between transmitted signals, it depends only on the two parameters of interest: power and effective bandwidth (see \cite{Gabor46} for the definition of effective bandwidth). It is shown in \cite{Godrich10-2} that this CRLB is the product of two factors, which capture the nature of the localization error: one is the CRLB of time delay estimation by a single sensor, and the second is a term known as Geometric Dilution of Precision (GDOP), which depends on the number of transmit and receive elements and their locations. Tight bounds are also available for lower SNR, such as the Barankin \cite{Tabrikian08} and Ziv-Zakai \cite{Musso07} bounds, however, they entail more complicated expressions not suitable to serve as cost functions.

In a configuration with multiple transmitters, the emitted signals may easily interfere with each other, hence the implementation of such techniques relies on orthogonality or low crosscorrelation of the received waveformss. This enables to separate incoming signals from different transmitters. One way to achieve orthogonality is by assigning different, disjoint transmission bandwidths to the active transmitters, as in a traditional netted radar \cite{Chernyak98}. In telecommunications field, this method is known as Frequency-Division Multiple Access (FDMA). Other schemes for multiple waveforms with low auto- and cross-correlations are proposed in \cite{Deng04,Hao09,Song10}. 
FDMA has some advantages over other multiple-emitter schemes: signal separation relies on simple filtering and the bandwidths of the transmitters can be flexibly changed without breaking the orthogonality of the signals.

In this paper, we are asking what are the optimal power allocation, bandwidth allocation and joint power bandwidth allocation to minimize the CRLB on the localization of multiple targets. In our model a fusion center performs resource allocations based on past localization estimates. Power allocation methods for localizing a single target can be found in \cite{Godrich11,Shen12}. The present work extends these results considerably by adding multiple targets, bandwidth allocation and joint power-bandwidth allocation. Initial results on the material contained in this paper can be found in our conference publication \cite{Garcia12}.

Power and/or bandwidth allocation problems are formulated as constrained optimizations, where we seek to minimize a function based on the CRLB, given constraints on the total power and bandwidth. These problems turn out to be non-convex, hence efficient techniques for solving convex problems do not apply. Instead we opt for an approach based on Sequential Parametric Convex Approximation (SPCA) \cite{Beck09}. In SPCA are devised a sequence of convexified problems that lead to an approximate solution of the allocation problem.

The main contributions of the paper are:
\begin{description}
	\item[a)] Formulation of a new unified framework for three resource allocation problems:
	optimal power allocation for a fixed bandwidth, optimal bandwidth allocation for a fixed total power, and optimal joint power and bandwidth allocations.
	
	\item[b)] Development of efficient algorithms for solving power and/or bandwidth allocation problems for static targets.
	Initial estimates of the targets locations are used as inputs to the optimization algorithms. The algorithms rely on SPCA framework in which a series of iterations pass to each other solutions of convexified allocation problems
	
	\item[c)] Obtaining lower-bounds on the optimum value of the cost function of the optimization problems. The lower-bounds provide a certificate that can confirm that the allocations resulting from our algorithms are close to the optimal one.
	
\end{description}

The paper is organized as follows. Section~\ref{sec:model} explains the signal model. Section~\ref{sec:problems} presents the different resource allocation problems considered in this article and reformulates them so that we only need to solve one problem. Section \ref{sec:solution} is devoted to finding an approximate solution and section~\ref{sec:optimality} provides a tool for assessing the quality of the solution.  Simulations and analysis of the solutions are provided in sections~\ref{sec:simulations}.


\section{Signal model} \label{sec:model}

Consider a MIMO-radar network consisting of $M$ transmitters, located
at $\{(x^{\textrm{tx}}_{m},y^\textrm{tx}_m)\}_{m=1}^{M}$, $N$ receivers, located at $\{(x^\textrm{rx}_n,y^\textrm{rx}_n)\}_{n=1}^{N}$, and $Q$ stationary targets, located at $\{(x^\textrm{tar}_q,y^\textrm{tar}_q)\}_{q=1}^{Q}$. Denote $d^\textrm{tx}_{m,q}$ and $d^\textrm{rx}_{q,n}$ the Euclidean distances from transmitter $m$ to target $q$ and
from target $q$ to receiver $n$, respectively. Let $\{s_m(t)\}_{m=1}^M$ be the transmitted pulses, where each pulse $s_m(t)$ has bandwidth $w_m$, energy $E_m$ and time duration $T_m$. Assuming a fixed pulse repetition $f_r$ frequency, the pulse energy $E_m$ is connected to the average power $p_m$ through the formula $p_m = E_m f_r$. For later use, we define the power vector $\mathbf{p}=[p_1,\ldots,p_M]^\top$ and the effective bandwidths vector $\mathbf{w}=[w_1,\ldots,w_M]^\top$. The transmitted signals are assumed narrowband in the sense that a target's frequency response (for a given transmitter-receiver pair) is represented by a complex-valued scalar. For sufficiently spaced sensors, the target returns vary among pairs, thus each target is modeled as a collection of $MN$ reflection coefficients. In this work, target returns are assumed deterministic and unknown. The low-pass signal observed at the $n$-th receiver is written as
\begin{equation}
	r_{n}(t) =\sum\limits_{q=1}^{Q}\sum\limits_{m=1}^{M} \sqrt{\alpha_{mqn}E_m}h_{mqn}s_{m}(t-\tau_{mqn})+e_{n}(t) \label{eq:signal}
\end{equation}
where $f_{c}$ is the carrier frequency, $c$ the speed of light,
\begin{equation}
	\alpha_{mqn}=\frac{1}{4\pi {d^{\textrm{tx}}_{m,q}}^2} \frac{1}{4\pi {d^{\textrm{rx}}_{q,n}}^2} \frac{1}{4\pi f_{c}^{2}}
\end{equation}
models the pathloss along the path transmitter $m$ -- target $q$ -- receiver $n$. The time delay along the path is $\tau_{mqn}$, and $h_{mqn}$ represents the targets complex gains. The signal propagation is assumed to occur in free-space, and the noise $e_n(t)$ is assumed Gaussian and white (AWGN) with constant power spectral density $N_0$.

The unknown parameters in \eqref{eq:signal} are the target locations
$\{(x^\textrm{tar}_q,y^\textrm{tar}_q)\}_{q=1}^{Q}$ and the $MQN$ complex gains $h_{mqn}$. The goal of the
radar system is to estimate the target location, with the complex gains $h_{mqn}$ serving as nuisance parameters. Our objective is to allocate resources (power and/or bandwidth)
to system elements, to optimize the target localization performance, using the CRLB as the optimization metric.

In the case of a single target, the CRLB is a $2\times2$ matrix, obtained by inverting the Fisher information matrix (FIM), whose diagonal elements are the lower-bounds on the variances of respectively the target location estimate along the x-axis ($\operatorname{var}(\hat{x}^\textrm{tar})$) and y-axis ($\operatorname{var}(\hat{y}^\textrm{tar})$). The trace of this matrix represents a lower-bound on the mean square error (MSE) of the target location estimate, i.e. $\operatorname{var}(\hat{x}^\textrm{tar}_q)+\operatorname{var}(\hat{y}^\textrm{tar}_q) \geq \operatorname{tr}\{\mathbf{C}_q\} \triangleq T_q$, where $\mathbf{C}_q$ is the CRLB matrix of a target located at $\{(x^\textrm{tar}_q,y^\textrm{tar}_q)\}$. An expression for the trace of the CRLB for localizing a single target using a single observation is derived in \cite{Godrich10}:
\begin{equation} \label{eq:CRLB_trace}
	\operatorname{var}(\hat{x}^\textrm{tar}_q)+\operatorname{var}(\hat{y}^\textrm{tar}_q)
	\geq
	T_q
	=\frac{\left(\mathbf{a}_q+\mathbf{b}_q\right)^\top \left(\operatorname{diag}\mathbf{w}\right)^2\mathbf{p}}
	 {\mathbf{p}^{\top}\left(\operatorname{diag}\mathbf{w}\right)^2\mathbf{H}_q\left(\operatorname{diag}\mathbf{w}\right)^2\mathbf{p}}
\end{equation}
The symbol $(\cdot)^\top$ denotes the transpose operator, $H_q = (0.5\mathbf{a}_q\mathbf{b}_q^\top +0.5\mathbf{b}_q\mathbf{a}_q^\top -\mathbf{c}_q\mathbf{c}_q^\top)$, and $\mathbf{a_q}$, $\mathbf{b}_q$ and $\mathbf{c}_q$
are length $M$ vectors defined as follows:
\begin{gather}
	\mathbf{a}_q=\eta
	\left[
	\begin{gathered}
		\sum\limits_{n=1}^{N}\alpha_{1qn}|h_{1qn}|^{2}\left(  \frac{x_1^\textrm{tx}-x_q^\textrm{tar}}{d_{1,q}^\textrm{tx}}+\frac{x_{n}^\textrm{rx}-x_q^\textrm{tar}}{d_{q,n}^\textrm{rx}}\right) ^{2}\\
		\vdots\\
		\sum\limits_{n=1}^{N}\alpha_{Mqn}|h_{Mqn}|^{2}
		\left(  \frac{x_M^\textrm{tx}-y_q^\textrm{tar}}{d_{M,q}^\textrm{tx}}+\frac{y_{n}^\textrm{rx}-y_q^\textrm{tar}}{d_{q,n}^\textrm{rx}}\right) ^{2}
	\end{gathered}
	\right]
	\displaybreak[2]\label{eq:a}\\
	\mathbf{b}_q=\eta
	\left[
	\begin{gathered}
		\sum\limits_{n=1}^{N}\alpha_{1qn}|h_{1qn}|^{2}\left(  \frac{y_1^\textrm{tx}-x_q^\textrm{tar}}{d_{1,q}^\textrm{tx}}+\frac{x_{n}^\textrm{rx}-x_q^\textrm{tar}}{d_{q,n}^\textrm{rx}}\right) ^{2}\\
		\vdots\\
		\sum\limits_{n=1}^{N}\alpha_{Mqn}|h_{Mqn}|^{2}\left(  \frac{x_M^\textrm{tx}-x_q^\textrm{tar}}{d_{M,q}^\textrm{tx}}+\frac{x_{n}^\textrm{rx}-x_q^\textrm{tar}}{d_{q,n}^\textrm{rx}}\right) ^{2}
	\end{gathered}
	\right]
	\displaybreak[2]\label{eq:b}\\
	\mathbf{c}_q=\eta
	\begin{bmatrix}
		\scriptstyle	
		\sum\limits_{n=1}^{N}\alpha_{1qn}|h_{1qn}|^{2}
		\left(  \frac{x_1^\textrm{tx}-x_q^\textrm{tar}}{d_{1,q}^\textrm{tx}}+\frac{x_{n}^\textrm{rx}-x_q^\textrm{tar}}{d_{q,n}^\textrm{rx}}\right)
		 \left(\frac{y_1^\textrm{tx}-x_q^\textrm{tar}}{d_{1,q}^\textrm{tx}}+\frac{x_{n}^\textrm{rx}-x_q^\textrm{tar}}{d_{q,n}^\textrm{rx}}\right)\\
		\vdots\\
		\scriptstyle	
		\sum\limits_{n=1}^{N}\alpha_{Mqn}|h_{Mqn}|^{2}
		 \left(\frac{x_M^\textrm{tx}-y_q^\textrm{tar}}{d_{M,q}^\textrm{tx}}+\frac{y_{n}^\textrm{rx}-y_q^\textrm{tar}}{d_{q,n}^\textrm{rx}}\right)
		 \left(\frac{x_M^\textrm{tx}-x_q^\textrm{tar}}{d_{M,q}^\textrm{tx}}+\frac{x_{n}^\textrm{rx}-x_q^\textrm{tar}}{d_{q,n}^\textrm{rx}}\right)
	\end{bmatrix}
	\label{eq:c}
\end{gather}
The constant $\eta$ is given by $\eta = \frac{8\pi^2}{c^2 f_r N_0}$. The vectors \eqref{eq:b} relate the CRLB parameter $T_q$ to the sensors locations, target location, target gain, and pathloss. Note that the matrix $\mathbf{H}_q$ is symmetric.

Computation of the CRLB for localizing $Q$ targets requires inverting a $2Q\times 2Q$ FIM which is a complicated mathematical operation. To simplify the matrix inversion, we make the following assumption which makes the FIM approximately block diagonal (see \cite{Wei10}):
\begin{equation} \label{eq:assumption}
	\int_{-\infty}^{\infty} s_m(t-\tau_{mqn}) s_m^*(t-\tau_{mq'n}) \,dt \approx 0
\end{equation}
for any pair of targets $q\neq q'$, and for any transmitter $m$ and receiver $n$. Let $\mathbf{C}$ denote the CRLB matrix for multiple targets, then assuming \eqref{eq:assumption}, the FIM becomes a block diagonal matrix, that can be easily inversed, and leads to a CRLB matrix for multiple targets $\overline{\mathbf{C}}$ whose main diagonal is
\begin{equation} \label{eq:diag_C}
	\operatorname{diag}\left(\overline{\mathbf{C}}\right) \approx \left[\operatorname{diag}\left(\mathbf{C}_1\right),\ldots,\operatorname{diag}\left(\mathbf{C}_Q\right)\right] \triangleq \operatorname{diag}\left(\mathbf{C}\right)
\end{equation}
where operator $\operatorname{diag}$ takes the main diagonal of the matrix between the brackets. The lower bounds on the variances of the targets locations estimates are obtained by taking the sum of the diagonal elements corresponding to the respective target
\begin{equation} \label{eq:var_lb}
	\operatorname{var}(\hat{x}^\textrm{tar}_q)+\operatorname{var}(\hat{y}^\textrm{tar}_q)
	\geq
	\operatorname{diag}\left(\overline{\mathbf{C}}\right)_{2(q-1)+1} +
	\operatorname{diag}\left(\overline{\mathbf{C}}\right)_{2q}
	\triangleq \overline{T}_q
\end{equation}
where the subindex selects the components in $\operatorname{diag}\left(\overline{\mathbf{C}}\right)$. Hence, if \eqref{eq:assumption} is satisfied, combining \eqref{eq:diag_C} and \eqref{eq:var_lb} results in $\overline{T}_q \approx T_q$. In section~\ref{sec:simulations}, the validity of this approximation will be assessed by evaluating $\overline{T}_q$ and $T_q$ for multiple scenarios and allocations.

\section{Resource allocation} \label{sec:problems}

\subsection{Problem formulation} \label{sub:formulation}

In this section we formulate several optimization problems that minimize the lower-bounds on the MSE (which are tight under a high-SNR regime) of the targets locations \eqref{eq:CRLB_trace} given constraints on power and bandwidth. These lower-bounds form a length $Q$ vector function $[\operatorname{C}_1(\mathbf{p},\mathbf{w}),\ldots,\operatorname{C}_Q(\mathbf{p},\mathbf{w})]$, where the dependency on the transmitters powers and bandwidths is made explicit. A standard technique for minimizing a vector function is known as scalarization \cite{Boyd}, and it consists in minimizing a scalar function whose input is the vector function.
In this regard a common requirement is ensuring that the localization error of any target is not too large. This can be cast as minimizing the worst CRLB on the variance of all target locations estimates, a criterion known in the literature as minimax. The worst MSE among all targets is written as the scalar function $\max_q \overline{T}_q$. According to this, a possible objective function $\overline{f}$ to minimize is
\begin{equation} \label{eq:true_cost_function}
	\overline{f}(\mathbf{p},\mathbf{w}) = \max_{q\in\{1,\ldots,Q\}} \overline{T}_q (\mathbf{p},\mathbf{w})
\end{equation}
We refer from now on to this cost function as the ``maximum CRLB''. However, because $\overline{T}_q$ lacks an easy-manipulable algebraic expression, we rely to its approximation
\begin{equation} \label{eq:cost_function}
	f(\mathbf{p},\mathbf{w}) = \max_{q\in\{1,\ldots,Q\}} T_q (\mathbf{p},\mathbf{w})
\end{equation}
which will be referred as the ``approximate maximum CRLB''.
Three allocation problems can now be formally stated using this objective function:

\begin{problem}[Power allocation] \label{prob:power} Given a total power $P$ and a fixed bandwidth $w$ per transmitter, the optimal power allocation is the solution $\mathbf{p}_{opt}$ to
	\begin{equation} \label{eq:power}
		\mathbf{p}_{opt} = \left\{
		\begin{aligned}
			\min_{\mathbf{p}}\quad &
			\max_{q\in\{1,\ldots,Q\}} T_q (\mathbf{p},w\mathbf{1}) \\
			\textnormal{s.t.} \quad &
			\mathbf{1}^\top\mathbf{p} \leq P \\
			& \mathbf{p}\geq\mathbf{0}
		\end{aligned} \right.
	\end{equation}
\end{problem}
\noindent where $\min_\mathbf{p}$ means ``minimize with respect to $\mathbf{p}$'', s.t.\ is the abbreviation for ``subject to'', $\mathbf{1}$ and $\mathbf{0}$ are the all-ones and all-zeros vectors respectively and $\leq$ is component-wise ``smaller or equal than''.

\begin{problem}[Bandwidth allocation] \label{prob:bandwidth} Given a uniform power allocation $p$ and a total bandwidth $B$, the bandwidth allocation is the solution $\mathbf{w}_{opt}$ to
	\begin{equation}  \label{eq:bandwidth}
		\mathbf{w}_{opt} = \left\{
		\begin{aligned}
			\min_{\mathbf{w}}\quad &
			\max_{q\in\{1,\ldots,Q\}} T_q (p\mathbf{1},\mathbf{w}) \\
			\textnormal{s.t.} \quad &
			\mathbf{1}^\top\mathbf{w} \leq B \\
			& \mathbf{w}\geq\mathbf{0}
		\end{aligned} \right.
	\end{equation}
\end{problem}

\begin{problem}[Joint power and bandwidth allocation] \label{prob:joint} Given a total power $P$ and bandwidth $B$, the power and bandwidth allocations are the solution $(\mathbf{p}_{opt}^j,\mathbf{w}_{opt}^j)$ to
	\begin{equation} \label{eq:joint}
		\left(\mathbf{p}_{opt}^j,\mathbf{w}_{opt}^j\right) = \left\{
		\begin{aligned}
			\min_{\mathbf{p},\mathbf{w}}\quad &
			\max_{q\in\{1,\ldots,Q\}} T_q (\mathbf{p},\mathbf{w}) \\
			\textnormal{s.t.} \quad &
			\mathbf{1}^\top\mathbf{p} \leq P \\
			& \mathbf{1}^\top\mathbf{w} \leq B  \\
			& \mathbf{p},\mathbf{w}\geq\mathbf{0}
		\end{aligned}\right.
	\end{equation}
\end{problem}

In all of the above problems, we minimize function \eqref{eq:cost_function}, which depends on the targets locations through $T_q$ \eqref{eq:CRLB_trace}. The target locations are not known, otherwise we would not allocate power and/or bandwidth to improve the localization accuracy. Thus from now on we assume that the target locations in $T_q$ are coarse estimates obtained in previous cycles and denoted $[\tilde{x}^\textrm{tar}_1,\tilde{y}^\textrm{tar}_1,\ldots,\tilde{x}^\textrm{tar}_Q,\tilde{y}^\textrm{tar}_Q]$.

\subsection{Unified framework} \label{sub:reformulation}

Next, we rewrite the three previous allocation problems in a unified form. Such reformulation enables us to find an approximate solution to all three problems using the same mathematical tools. In order to write the allocation problems in this unified form, we first need to perform some algebraic manipulations. We start by enunciating the following lemma:
\begin{lemma}[Scaling] \label{lem:scale}
	For any constants $\alpha,\beta> 0$, the lower-bound on the MSE of a target location meets the following scaling property
	\begin{equation}
		T_q (\alpha\mathbf{p},\beta\mathbf{w})=\frac{1}{\alpha\beta^2}T_q(\mathbf{p},\mathbf{w})
	\end{equation}
\end{lemma}
\begin{proof}
	It suffices to expand $T_q (\alpha\mathbf{p},\beta\mathbf{w})$ using \eqref{eq:CRLB_trace}.
\end{proof}
Applying Lemma~\ref{lem:scale} to Problems~\ref{prob:power} and \ref{prob:bandwidth}, simplifies the respective objective functions by moving from scaling the argument to scaling the full function. The following proposition will be applied to simplify Problem~\ref{prob:joint}:
\begin{proposition} \label{prop:proportionallity}
	The power $\mathbf{p}_{opt}^j$ and bandwidth $\mathbf{w}_{opt}^j$ solutions to Problem \ref{prob:joint} are related through
	\begin{equation}
		\mathbf{w}_{opt}^j = \frac{B}{P} \mathbf{p}_{opt}^j
	\end{equation}
\end{proposition}

\begin{proof}
	See Appendix~\ref{appen:proportionallity}.
\end{proof}

Knowing beforehand how the power $(\mathbf{p}_{opt}^j)$ and bandwidth $(\mathbf{w}_{opt}^j)$ solutions relate to each other for Problem~\ref{prob:joint}, we can restrict the feasible set of Problem~\ref{prob:joint} to points satisfying $\mathbf{w}=\frac{B}{P}\mathbf{p}$. Performing such substitution for vector $\mathbf{w}$ in Problem~\ref{prob:joint}, transforms it into an optimization problem with only one vector variable $\mathbf{p}$ instead of two,
\begin{equation} \label{eq:joint_simplified}
	\mathbf{p}_{opt}^j = \left\{
	\begin{aligned}
		\min_{\mathbf{p}}\quad &
		\max_{q\in\{1,\ldots,Q\}} T_q (\mathbf{p},\mathbf{p}) \\
		\text{s.t.} \quad &
		\mathbf{1}^\top\mathbf{p} \leq P \\
		& \mathbf{p}\geq\mathbf{0}
	\end{aligned}
	\right.
\end{equation}
where $\mathbf{w}_{opt}^j$ is recovered by applying Proposition~\ref{prop:proportionallity}.

At this point, it has been shown that the specific values of bandwidth per user ($w$) and power per user ($p$) in Problems~\ref{prob:power} and \ref{prob:bandwidth}, do not change the solutions, and consequently Problems~\ref{prob:power} and \ref{prob:bandwidth} can be solved with $w=p=1$. Additionally thanks to Proposition~\ref{prop:proportionallity}, Problem~\ref{prob:joint} has been rewritten in a more compact form \eqref{eq:joint_simplified}. Notice that except for the problem specific constants $P$ and $B$,  Problems~\ref{prob:power}, \ref{prob:bandwidth} and \ref{prob:joint} (expressed as \eqref{eq:joint_simplified}) have similar objective functions and constraints. In fact, upon introducing the function
\begin{equation} \label{eq:x}
	g_q(\mathbf{y},k)
	=\frac{\left(\mathbf{a}_q+\mathbf{b}_q\right)^\top \left(\operatorname{diag}\mathbf{y}\right)^k\mathbf{y}}
	 {\mathbf{y}^{\top}\left(\operatorname{diag}\mathbf{y}\right)^k\mathbf{H}_q\left(\operatorname{diag}\mathbf{y}\right)^k\mathbf{y}}
\end{equation}
that closely resembles that of $T_q$ in \eqref{eq:CRLB_trace}, the three problems can be dealt with in a unified manner through
\begin{equation} \label{eq:unified}
	\begin{aligned}
		\min_{\mathbf{y}}\quad & \max_{q\in\{1,\ldots,Q\}} g_q (\mathbf{y},k) \\
		\textnormal{s.t.} \quad &
		\mathbf{1}^\top\mathbf{y} \leq D  \\
		& \mathbf{y}\geq\mathbf{0}
	\end{aligned}
\end{equation}
Denoting the optimal solution $\mathbf{y}_{opt}$, it is easily verified that:
\begin{itemize}
	\begin{subequations} \label{eq:unified_bullets}
		\item Let
		\begin{align} \label{eq:unified_power}
			\mathbf{y}=\mathbf{p} &&& D=P && k=0
		\end{align}
		then $\mathbf{p}_{opt}=\mathbf{y}_{opt}$
	
		\item Let
		\begin{align} \label{eq:unified_bandwidth}
			\mathbf{y}=\mathbf{w} &&& D=B && k=1
		\end{align}
		then $\mathbf{w}_{opt}=\mathbf{y}_{opt}$
		
		\item Let
		\begin{align} \label{eq:unified_joint}
			\mathbf{y}=\mathbf{p} &&& D=P && k=2
		\end{align}
	\end{subequations}
	then $(\mathbf{p}_{opt}^j,\mathbf{w}_{opt}^j)=(\mathbf{y}_{opt},\frac{B}{P}\mathbf{y}_{opt})$
\end{itemize}

{\tiny \color{white} \eqref{eq:unified_power}\eqref{eq:unified_bandwidth}\eqref{eq:unified_joint}}

An optimization problem where the $\max$ operator appears in the constraints, instead of being in the objective function like in problem \eqref{eq:unified}, is in general easier to solve. For instance if a problem had constraint $\max_{q\in\{1,\ldots,Q\}} g_q (\mathbf{y},k)\leq 1$, then the $\max$ operator may be avoided by expressing the constraint
\begin{equation} \label{eq:x_constraints}
	g_q(\mathbf{y},k) \leq 1 \qquad\textrm{for all } q=1,\ldots,Q
\end{equation}
It turns out that the solution to the optimization problem,
\begin{equation} \label{eq:x_reformulated}
	\begin{aligned}
		\min_{\mathbf{y}}\quad & \mathbf{1}^\top\mathbf{y} \\
		\text{s.t.} \quad &
		\max_{q\in\{1,\ldots,Q\}} g_q(\mathbf{y},k) \leq 1 \\
		& \mathbf{y}\geq\mathbf{0}
	\end{aligned}
\end{equation}
which has no $\max$ operator in the objective function, can be closely related to that of problem (\ref{eq:unified}) as established by the following proposition:
\begin{proposition} \label{prop:power_dual}
	The solution to problem \eqref{eq:unified} is the same than for problem \eqref{eq:x_reformulated} except for a scaling factor.
\end{proposition}

\begin{proof}
	See Appendix~\ref{appen:equivalent}.
\end{proof}

Writing the constraints of problem \eqref{eq:x_reformulated} in the form of \eqref{eq:x_constraints}, and expressing $g_q(\mathbf{y},k)$ using \eqref{eq:x}, problem \eqref{eq:x_reformulated} is rewritten as
\begin{problemNoDots}[Canonical problem] \label{prob:unified2}
	\begin{align}
		\min_{\mathbf{y}}\quad & \mathbf{1}^\top\mathbf{y} \\
		\textnormal{s.t.} \quad &
		\left(\mathbf{a}_q+\mathbf{b}_q\right)^\top \left(\operatorname{diag}\mathbf{y}\right)^k\mathbf{y}
		\leq
		 \mathbf{y}^{\top}\left(\operatorname{diag}\mathbf{y}\right)^k\mathbf{H}_q\left(\operatorname{diag}\mathbf{y}\right)^k\mathbf{y} \\[-3ex]
		& \qquad\qquad\qquad\qquad\qquad\qquad\textnormal{for all } q=1,\ldots,Q \label{eq:unified2_polynomials}\\
		& \mathbf{y}\geq\mathbf{0}
	\end{align}
\end{problemNoDots}

Problem~\ref{prob:unified2} is the final form we seek to attain in this section, however, the scaling constant relating its solution to the solution of problem \eqref{eq:unified} is still unknown. For this purpose we enunciate the following lemma

\begin{lemma} \label{lem:active_unified}
	Constraint $\mathbf{1}^\top\mathbf{y}\leq D$ in problem \eqref{eq:unified} is active at the solution $\mathbf{y}_{opt}$.
\end{lemma}

\begin{proof}
	The proof is done by contradiction. Assume there exists a minimum point $\mathbf{y}'$ such that $\mathbf{1}^\top\mathbf{y}' < D$, i.e.\ constraint is not active. Define $\mathbf{y}^*=\frac{D}{\mathbf{1}^\top\mathbf{y}'}\mathbf{y}'$. From the definition of $g_q(\mathbf{y},k)$ in \eqref{eq:x}
	\begin{equation}
		\max_{q} g_q(\mathbf{y}^*,k) =
		\left(\frac{\mathbf{1}^\top\mathbf{y}'}{D}\right)^{k+1} \max_{q} g_q(\mathbf{y}',k)
		< \max_{q} g_q(\mathbf{y}',k)
	\end{equation}
	This contradicts the assumption that the CRLB achieves its minimum at $\mathbf{y}'$. It follows that the power constraint evaluated at the optimal point must be active.
\end{proof}

Thus by Lemma~\ref{lem:active_unified} and Proposition~\ref{prop:proportionallity}, if $\mathbf{y}'$ is solution of problem \eqref{eq:unified}, then a solution to Problem~\ref{prob:unified2} is given by
\begin{equation} \label{eq:relation_power}
	\mathbf{y}_\textrm{opt} = \frac{D}{\mathbf{1}^\top\mathbf{y}'} \mathbf{y}'
\end{equation}

Using this result we can now solve Problem~\ref{prob:unified2} and relate it to the solution of problem \eqref{eq:unified}, which relates in turn, is related to the original Problems~\ref{prob:power}-\ref{prob:joint} via \eqref{eq:unified_bullets}. Putting it all together, we obtain the direct link between Problem~\ref{prob:unified2} which we plan to solve and the original Problems~\ref{prob:power}-\ref{prob:joint}. Given the solution to Problem~\ref{prob:unified2}, $\mathbf{y}'$, the solutions to Problems~\ref{prob:power}, \ref{prob:bandwidth} or \ref{prob:joint}, depending on the value of $k$, can be obtained as follows:
\begin{equation} \label{eq:relation}
	\begin{split}
		\begin{aligned}
			&\textrm{Problem}\textrm{ \ref{prob:power}} & &\textrm{Problem~\ref{prob:bandwidth}} \\
			\mathbf{p}_{opt} &= \left.\frac{P}{\mathbf{1}^\top\mathbf{y}'}\mathbf{y}'\right|_{k=0} &
			\mathbf{w}_{opt} &= \left.\frac{B}{\mathbf{1}^\top\mathbf{y}'}\mathbf{y}'\right|_{k=1}
		\end{aligned} \\[2ex]
		\begin{gathered}
			\textrm{Problem}\textrm{ \ref{prob:joint}} \\
			\mathbf{p}_{opt}^j = \left.\frac{P}{\mathbf{1}^\top\mathbf{y}'}\mathbf{y}'\right|_{k=2} \qquad
			\mathbf{w}_{opt}^j = \left.\frac{B}{\mathbf{1}^\top\mathbf{y}'}\mathbf{y}'\right|_{k=2}
		\end{gathered}
	\end{split}
\end{equation}

\section{Proposed approximate solution} \label{sec:solution}

Currently our goal  is to find an approximate solution to Problem~\ref{prob:unified2} for any $k\in\{0,1,2\}$. It is not difficult to see that for Problem~\ref{prob:unified2}, the objective function is linear, constraints \eqref{eq:unified2_polynomials} are polynomials of 2nd (case $k=0$), 4th ($k=1$) or 6th ($k=2$) order, and the last constraint is linear. The problem would fit in the framework of convex optimization, for which very efficient techniques exist \cite{Boyd}, if \eqref{eq:unified2_polynomials} were convex. However, this is true only for some very particular cases of $\mathbf{H}_q$. Therefore we have to rely on techniques designed for nonconvex optimization. In most cases such techniques do not lead to the global minimum. Moreover, the computational cost of such techniques grows exponentially  with the dimension of the problem.

An alternative approach is to approximate the original problem with a sequence of convex problems. An approach firstly proposed in \cite{Marks78}, and later referred to as Sequential Parametric Convex Approximation (SPCA) \cite{Beck09}. To explain SPCA, we start by observing that any constraint in Problem~\ref{prob:unified2} may be written as $h(\mathbf{y})\leq 0$, where $h$ will be called the constraint function. If $h$ is a convex function, then $h(\mathbf{y})\leq 0$ is a convex constraint.  The main idea of SPCA is that at each iteration, each of the nonconvex constraints functions in Problem~\ref{prob:unified2} is replaced by a convex approximation. To construct such approximation, the nonconvex function is decomposed into a sum of a convex and a concave function. The concave function is linearized around a point as proposed in \cite{Aspremont03}. At each iteration, the algorithm solves the approximate convex problem. It stops when there is no further improvement in the objective function. The solution at each iteration is passed to the next iteration as the linearization point. Convergence of this algorithm is ensured at least to a local minimum \cite{Marks78}.

As mentioned previously, the first step is to decompose the nonconvex constraints functions in Problem~\ref{prob:unified2} into a sum of a convex and a concave function. Accomplishing this for \eqref{eq:unified2_polynomials} is not straightforward. First, we introduce a vector of slack variables $\mathbf{z} = [z_1,\ldots,z_M]^\top$, which are linked to the components of $\mathbf{y}$ by $z_m=y_m^{k+1}$. Problem~\ref{prob:unified2} is recast as
\begin{subequations} \label{eq:slack_problem}
	\begin{alignat}{3}
		\min_{\mathbf{y},\mathbf{z}} & & &\mathbf{1}^\top\mathbf{y} \\
		\textnormal{s.t.} & & \; & \left(\mathbf{a}_q+\mathbf{b}_q\right)^\top \mathbf{z} - \mathbf{z}^{\top}\mathbf{H}_q\mathbf{z} \leq 0
		& \;\;\;\textnormal{for }q=1,\ldots,Q& \label{eq:slack_polynomial}\\
		& & & z_m - y_m^{k+1} = 0 & \hspace{-1em} \textrm{for }m=1,\ldots,M& \hspace{1.3em} \label{eq:slack_equality} \\
		& & &\mathbf{y},\mathbf{z}\geq\mathbf{0}
	\end{alignat}
\end{subequations}

The advantage of this new problem over Problem~\ref{prob:unified2} is that \eqref{eq:slack_polynomial} is now in a quadratic form, and a simple way to decompose a quadratic function into a convex plus concave function is to separate matrix $\mathbf{H}_q$ into a sum of a nonnegative-definite ($\mathbf{H}^+_q$) and a nonpositive-definite ($\mathbf{H}^-_q$) matrix. Concerning the newly introduced constraint \eqref{eq:slack_equality}, notice that if we relax it by putting $z_m - y_m^{k+1} \leq 0$ instead, then it is also a sum of a convex ($z_m$) and a concave ($-y_m^{k+1}$) function:
\begin{subequations} \label{eq:DCproblem}
	\begin{align}
		\min_{\mathbf{y},\mathbf{z}}\quad & \mathbf{1}^\top\mathbf{y} \\
		\textnormal{s.t.} \quad &
		\left(\mathbf{a}_q+\mathbf{b}_q\right)^\top \mathbf{z} - \mathbf{z}^{\top}\mathbf{H}_q^-\mathbf{z} - \mathbf{z}^{\top}\mathbf{H}_q^+\mathbf{z} \leq 0 \label{eq:DCproblem_polynomial} \\
		& \qquad\qquad\qquad\qquad\qquad\textnormal{for }q=1,\ldots,Q \\
		& z_m - y_m^{k+1} \leq 0 \quad\qquad\textrm{for }m=1,\ldots,M \label{eq:DCproblem_relax} \\
		& \mathbf{y},\mathbf{z}\geq\mathbf{0}
	\end{align}
\end{subequations}

\begin{lemma}
	The optimal solution to problem~\eqref{eq:DCproblem} always satisfies \eqref{eq:DCproblem_relax} with equality, and therefore it is also the optimal solution to problem~\eqref{eq:slack_problem}.	
\end{lemma}
\begin{proof}
	Assume that is a solution $(\mathbf{y}^{opt},\mathbf{z}^{opt})$ such that for some component $m$ satisfies $z_m^{opt}<(y_m^{opt})^{k+1}$. Then we can define another  point $(\mathbf{y}',\mathbf{z}^{opt})$ such that the $m$th component of $\mathbf{y}'$ is $y'_m=(z_m^{opt})^{\frac{1}{k+1}}<y_m^{opt}$ and the remaining components are $y'_m=y^{opt}_m$. This point satisfies $z_m - y_m^{k+1} \leq 0$ with equality and gives a smaller value for the objective function, contradicting the fact that $(\mathbf{y}^{opt},\mathbf{z}^{opt})$ is a solution.
\end{proof}

Given that problem \eqref{eq:DCproblem}'s constraints are separated into convex and concave functions, we can then convexify the problem by linearizing the concave parts, $-\mathbf{z}^{\top}\mathbf{H}_q^+\mathbf{z}$ in \eqref{eq:DCproblem_polynomial} and $- y_m^{k+1}$ in \eqref{eq:DCproblem_relax} around a point $(\mathbf{y}_{(n)},\mathbf{z}_{(n)})$. Linearization may be implemented by a first order Taylor expansion, where $n$ indexes the iteration. The optimization problem becomes then:
\begin{subequations} \label{eq:convexified}
	\begin{align}
		\min_{\mathbf{y},\mathbf{z}}\quad & \mathbf{1}^\top\mathbf{y} \label{eq:convexified_objective}\\
		\textnormal{s.t.} \quad &
		\left(\mathbf{a}_q+\mathbf{b}_q\right)^\top \mathbf{z} - \mathbf{z}^{\top}\mathbf{H}_q^-\mathbf{z} -\mathbf{z}_{(n)}^{\top}\mathbf{H}_q^+(2\mathbf{z}-\mathbf{z}_{(n)}) \leq 0 \\
		& \qquad\qquad\qquad\qquad\qquad\textnormal{for }q=1,\ldots,Q \label{eq:convexified_polynomial} \\
		& z_m + k y_{(n),m}^{k+1} -(k+1)y_{(n),m}^{k}x_{m} \leq 0 \label{eq:convexified_relax}  \\
		& \qquad\qquad\qquad\qquad\qquad\textrm{for }m=1,\ldots,M  \\
		& \mathbf{y},\mathbf{z}\geq\mathbf{0}
	\end{align}
\end{subequations}
The feasible set of problem~\eqref{eq:convexified} is convex, and in addition, it is a subset of the feasible set of problem~\eqref{eq:DCproblem}. To confirm this point, notice the constraint function in \eqref{eq:convexified_relax} is equal or larger than the constraint function in \eqref{eq:DCproblem_relax} for all $(\mathbf{y},\mathbf{z})$, and therefore, the set of points satisfying constraint \eqref{eq:convexified_relax} is a subset of the one defined by \eqref{eq:DCproblem_relax}. The set of points defined by \eqref{eq:convexified_relax} is also a subset of \eqref{eq:DCproblem_polynomial}. Therefore any solution resulting from solving the approximate problem \eqref{eq:convexified_polynomial}  is in the feasible set of problem \eqref{eq:DCproblem_polynomial}, and consequently of Problem~\ref{prob:unified2}.



\vspace{.8em}\noindent\textbf{Algorithm :} First, depending if we are allocating power (Problem~\ref{prob:power}), bandwidth (Problem~\ref{prob:bandwidth}) or joint power-bandwidth (Problem~\ref{prob:joint}) we set $k$ in \eqref{eq:convexified} to 0, 1 or 2 respectively. The algorithm consists in solving a series of convex problems \eqref{eq:convexified}. The solution $(\mathbf{y}_{(i)},\mathbf{z}_{(i)})$ to \eqref{eq:convexified} at each iteration $i$ is passed to the next iteration $i+1$ and used as a linearization point for \eqref{eq:convexified}. For the initialization step we choose the uniform allocation $\mathbf{y}_{(0)}, \mathbf{z}_{(0)} \propto \mathbf{1}$ as the linearization point because it treats all transmitters equally. The algorithm stops when the value in the cost function \eqref{eq:convexified_objective} does not change substantially. After denoting $\mathbf{y}'$ the solution to \eqref{eq:convexified} in the last iteration, the allocation vector is recovered via \eqref{eq:relation}.

\section{Lower bound on the accuracy of the optimal allocations}
\label{sec:optimality}

The previous section provided an algorithm that finds approximate solutions to Problems~\ref{prob:power} to \ref{prob:joint}. In this section, we provide a method for assessing their quality. We denote optimal solutions $\mathbf{p}_{opt}$, $\mathbf{w}_{opt}$ and $(\mathbf{p}_{opt}^j, \mathbf{w}_{opt}^j)$, and approximate solutions  $\tilde{\mathbf{p}}_{opt}$, $\tilde{\mathbf{w}}_{opt}$ and $(\tilde{\mathbf{p}}_{opt}^j, \tilde{\mathbf{w}}_{opt}^j)$ obtained using the algorithm in the previous section, to Problems~\ref{prob:power}, \ref{prob:bandwidth} and \ref{prob:joint} respectively. Obviously the objective function of Problems~\ref{prob:power}, \ref{prob:bandwidth} and \ref{prob:joint} evaluated at the approximate solutions, are equal or larger than if it they were evaluated at the optimal points:
\begin{gather}
	\max_{q}T_q(\tilde{\mathbf{p}}_{opt},w\mathbf{1}) \geq 	\max_{q}T_q(\mathbf{p}_{opt},w\mathbf{1}) \geq \mathrm{L}_p \label{eq:lb_power}\\
	\max_{q}T_q(p\mathbf{1},\tilde{\mathbf{w}}_{opt}) \geq		\max_{q}T_q(p\mathbf{1},\mathbf{w}_{opt}) \geq \mathrm{L}_b \\
	\max_{q}T_q(\tilde{\mathbf{p}}_{opt}^j,\tilde{\mathbf{w}}_{opt}^j) \geq		 \max_{q}T_q(\mathbf{p}_{opt}^j,\mathbf{w}_{opt}^j) \geq \mathrm{L}_j
\end{gather}
where $\mathrm{L}_p$, $\mathrm{L}_b$ and $\mathrm{L}_j$ are some lower-bounds on the unknown global minimums. This section is devoted to developing these lower-bounds, and their usefulness is explained in the following example. Assume that for Problem~\ref{prob:power} (power allocation) the lower-bound is tight to our approximate minimum $\max_{q}T_q(\tilde{\mathbf{p}}_{opt},w\mathbf{1}) \approx \mathrm{L}_p$, then by \eqref{eq:lb_power} we can conclude that our approximate minimum is very close to the global minimum $\max_{q}T_q(\tilde{\mathbf{p}}_{opt},w\mathbf{1}) \approx \max_{q}T_q(\mathbf{p}_{opt},w \mathbf{1}) \approx \mathrm{L}_p$. However, nothing can be asserted if $\max_{q}T_q(\tilde{\mathbf{p}}_{opt},w \mathbf{1}) \gg \mathrm{L}_p$.

Instead of finding lower-bounds for the global minimum of Problems~\ref{prob:power}, \ref{prob:bandwidth} and \ref{prob:joint} separately, we rely on the following proposition to simplify the process:
\begin{proposition} \label{prop:lower-bounds}
	Let $L_c$ denote a lower-bound to the global minimum of Problem~\ref{prob:unified2}. Then we can obtain lower-bounds, $\mathrm{L}_p$, $\mathrm{L}_b$ and $\mathrm{L}_j$, to the global minimums of Problems~\ref{prob:power}, \ref{prob:bandwidth} and \ref{prob:joint} respectively, through the following equations
	\begin{align} \label{eq:lb_relation}
		L_p &= \left. \frac{L_c}{P w^2}   \right|_{k=0}, &
		L_b &= \left. \frac{L_c^2}{p B^2}  \right|_{k=1}, &
		L_j &= \left. \frac{L_c^3}{P B^2}  \right|_{k=2}
	\end{align}
\end{proposition}
Therefore it suffices to find a lower-bound ($L_c$) for Problem~\ref{prob:unified2}. The following lemma is needed for the proof of Proposition~\ref{prop:lower-bounds}.

\begin{lemma} \label{lem:active_CRLB}
	Constraint \eqref{eq:unified2_polynomials} in  Problem~\ref{prob:unified2} must be active when evaluated at the solution.
\end{lemma}

\begin{proof}[Proof of Lemma~\ref{lem:active_CRLB}]
	Assume there exists a minimum point $\mathbf{y}'$ such that $\max_{q} g_q(\mathbf{y}',k) < 1$. Define $\mathbf{y}^*=[\max_{q} g_q (\mathbf{y}',k)]^\frac{1}{k+1} \mathbf{y}'$. From \eqref{eq:x} we can derive that $\max_q g_q(\mathbf{y}^*,k) =1$, thus it satisfies the constraints of the optimization problem \eqref{eq:x_reformulated}. Next
	\begin{equation}
		\mathbf{1}^\top\mathbf{y}^* =\left[\max_{q\in\{1,\ldots,Q\}} g_q (\mathbf{y}',k)\right]^\frac{1}{k+1} \mathbf{1}^\top\mathbf{y}'
		< \mathbf{1}^\top\mathbf{y}'
	\end{equation}
	This contradicts the assumption that the problem achieves its minimum at $\mathbf{y}'$. It follows that \eqref{eq:x_reformulated} evaluated at the optimal point must be active.
\end{proof}

\begin{proof}[Proof of Proposition~\ref{prop:lower-bounds}]
	The proof is done, only, for Problem~\ref{prob:power}'s lower-bound $L_p$. For Problems \ref{prob:bandwidth} and \ref{prob:joint}, the proof follows the same steps with very minor differences and are omitted here for brevity. Recall that, according to \eqref{eq:relation}, the solution of Problem~\ref{prob:power} is directly related to the solution of Problem~\ref{prob:unified2} by $\mathbf{p}_{opt} = \frac{P}{\mathbf{1}^\top\mathbf{y}'}\mathbf{y}'|_{k=0}$. Substituting $\mathbf{p}_{opt}$ in the global minimum \eqref{eq:lb_power} of Problem~\ref{prob:power} results in
	\begin{equation}
		\max_{q}T_q(\mathbf{p}_{opt},w\mathbf{1}) = \left.\frac{\mathbf{1}^\top\mathbf{y}'}{Pw^2} \max_{q}T_q(\mathbf{y}',\mathbf{1})\right|_{k=0}
	\end{equation}
	where we made used of Lemma~\ref{lem:scale} to simplify it. By Lemma~\ref{lem:active_CRLB} it can be farther reduced to $\max_{q}T_q(\mathbf{p}_{opt},w\mathbf{1})= \frac{\mathbf{1}^\top\mathbf{y}'}{Pw^2}|_{k=0}$. By definition $L_c$ is a lower-bound of the global minimum of Problem~\ref{prob:unified2}, so it must satisfy $L_c\leq \mathbf{1}^\top\mathbf{y}'$, which leads to
	\begin{equation}
		\max_{q}T_q(\mathbf{p}_{opt},w\mathbf{1}) \geq \left.\frac{L_c}{Pw^2} \right|_{k=0}
	\end{equation}
	The right side is obviously a lower-bound to global minimum of Problem~\ref{prob:power} and proves the first equality in Proposition~\ref{prop:lower-bounds}.
\end{proof}

To get a lower-bound $L_c$, we apply a series of relaxations on the feasible set of Problem~\ref{prob:unified2} in order to obtain another optimization problem whose solution can be computed, and whose global minimum is equal or smaller than that of Problem~\ref{prob:unified2}. This global minimum then constitutes a lower bound to the minimum of Problem~\ref{prob:unified2}, which will be denoted by $L_c$. To that end, the first step consists in making a variable vector substitution $z_m = y_m^{k+1}$ for all $m = 1,\ldots,M$ in Problem~\ref{prob:unified2}. Such operation does not change the global minimum of the problem.
\begin{align}
	\min_{\mathbf{z}}\quad & \sum_{m=1}^{M}\sqrt[k+1]{z_m} \label{eq:variable_change} \\
	\textnormal{s.t.} \quad &
	\left(\mathbf{a}_q+\mathbf{b}_q\right)^\top \mathbf{z} \leq 	\mathbf{z}^{\top}\mathbf{H}_q\mathbf{z} \quad\textnormal{for } q=1,\ldots,Q \label{eq:variable_change_poly} \\
	& \mathbf{z}\geq\mathbf{0}
\end{align}
Where $\mathbf{z}=[z_1,\ldots,z_M]^\top$. It is easy to verify that an equal or smaller objective function to that of \eqref{eq:variable_change} for all values of $z_m$ is
\begin{equation} \label{eq:new_objective}
	\sqrt[k+1]{\mathbf{1}^\top\mathbf{z}}
\end{equation}
Let $\mathbf{z}'$ denote the solution to problem \eqref{eq:variable_change_poly} with the new objective function \eqref{eq:new_objective} instead of \eqref{eq:variable_change}. The minimum will be equal or smaller than that of Problem~\ref{prob:unified2}, i.e. $\sqrt[k+1]{\mathbf{1}^\top\mathbf{z}'}\leq \mathbf{\mathbf{1}^\top\mathbf{y}'}$. As the root is a monotonically increasing function, suppressing it from the objective function \eqref{eq:new_objective} still leads to the same solution $\mathbf{z}'$, and therefore we can simplify Problem \eqref{eq:variable_change}-\eqref{eq:variable_change_poly} to
\begin{align}
	\min_{\mathbf{z}}\quad & \mathbf{1}^\top\mathbf{z} \\
	\textnormal{s.t.} \quad &
	\left(\mathbf{a}_q+\mathbf{b}_q\right)^\top \mathbf{z} \leq 	\mathbf{z}^{\top}\mathbf{H}_q\mathbf{z}
	\quad\textnormal{for } q=1,\ldots,Q  \label{eq:Zproblem_poly}\\
	& \mathbf{z}\geq\mathbf{0}
\end{align}
We relax now the feasible set of problem \eqref{eq:Zproblem_poly} by removing $Q-1$ constraints from \eqref{eq:Zproblem_poly}:
\begin{equation} \label{eq:power_st}
	\begin{aligned}
		\min_{\mathbf{z}}\quad & \mathbf{1}^\top\mathbf{z} \\
		\textnormal{s.t.} \quad &
		\left(\mathbf{a}_q+\mathbf{b}_q\right)^\top \mathbf{z} \leq
		\mathbf{z}^{\top}\mathbf{H}_q\mathbf{z}\\
		& \mathbf{z}\geq\mathbf{0}
	\end{aligned}
\end{equation}
Here $q$ takes only one value between 1 and $Q$. Call $\mathbf{z}'_q$ the solution to problem \eqref{eq:power_st}. Because problem \eqref{eq:power_st} is a relaxation of problem \eqref{eq:Zproblem_poly}, its minimum satisfies $\mathbf{1}^\top\mathbf{z}'_q\leq\mathbf{1}^\top\mathbf{z}'$. It turns out that this problem has the same algebraic form as the power allocation problem in \cite{Godrich11} (Section III.A.2), where an exact solution is provided by solving the Karush-Kuhn-Tucker conditions \cite{Bertsekas}. If we solve it for all possible values of $q\in\{1,\ldots,Q\}$, then we can obtain the tighter inequality $\max_q (\mathbf{1}^\top\mathbf{z}'_q) \leq \mathbf{1}^\top\mathbf{z}'$. Putting this together with the fact that $\sqrt[k+1]{\mathbf{1}^\top\mathbf{z}'}\leq \mathbf{\mathbf{1}^\top\mathbf{y}'}$, we obtain the desired computable lower-bound for Problem~\ref{prob:unified2}'s global minimum:
\begin{equation} \label{eq:lower_bound2}
	L_c = \sqrt[k+1]{\max_{q\in\{1,\ldots,Q\}} \mathbf{1}^\top\mathbf{z}'_q}
	\leq \mathbf{1}^\top\mathbf{y}'
\end{equation}

Combining \eqref{eq:lower_bound2} with \eqref{eq:lb_relation}, the final expressions for the lower-bounds of Problems~\ref{prob:power}, \ref{prob:bandwidth} and \ref{prob:joint} are
\begin{align}
	L_p &= \frac{\max_q \mathbf{1}^\top\mathbf{z}'_q}{P w^2}, &
	L_b &= \frac{\max_q \mathbf{1}^\top\mathbf{z}'_q}{p B^2}, &
	L_j &=\frac{\max_q \mathbf{1}^\top\mathbf{z}'_q}{P B^2}
\end{align}
Quite relevant is that if these lower-bounds are tight to the minimums of Problems~\ref{prob:power}, \ref{prob:bandwidth} and \ref{prob:joint}, it suggests that bandwidth has a bigger impact than power because the bandwidth variables $w$ and $B$ appear as quadratic terms in comparison to $p$ and $P$.


\section{Numerical results}
\label{sec:simulations}

The numerical examples presented in this section were obtained with five transmitters, five receivers and four targets. The choice of the number of elements enables sufficient choice for resource allocation, while not making the system overly complex. The total bandwidth available to the network is set to \SI{3}{\mega\hertz}. The average power available for the network is an adjustable parameter. The pulse repetition frequency is set to \SI{5}{\kilo\hertz}, which is sufficient for unambiguous range estimation in our setup. The targets are static. We choose a pulse integration time of \SI{10}{\milli\second}.

The proposed allocation algorithms assign power and/or bandwidth depending on the specific locations of the elements and the reflection coefficients of the targets. To avoid obtaining results that are specific to a particular layout, each point in the figures that are to follow is formed as averaging results of 1000 simulations. For each simulation, the transmitters, receivers and targets are positioned randomly in a \SI{20x20}{\kilo\metre} area, according to a uniform distribution. The reflection coefficients of the targets are drawn from a complex random variable with variance \SI{10}{\square\metre}. Performance was evaluated from the average of 1000 values of the cost function \eqref{eq:true_cost_function}. Each value represents an optimal allocation of power, bandwidth, or joint power-bandwidth for an instantiation of targets locations, reflection coefficients and noise.

\subsection{Resource allocation for different SNR values} \label{sub:simulations_snr}

Fig.~\ref{fig:crlb} presents the square root of the max CRLB as a function of the relative SNR for four different case studies: power and bandwidth evenly distributed among transmitters, power allocation, bandwidth allocation, and joint power-bandwidth allocation. As expected, the joint allocation performs the best decreasing the cost function by 70\% compared to uniform allocation, which has the worst performance. Bandwidth allocation is second best and power allocation is just slightly better than uniform allocation, with decreases in the cost of 50\% and 10\%, respectively. Increasing the SNR improves the localization accuracy for all methods.

\begin{figure}[ptb]
	\includegraphics[width=\columnwidth]{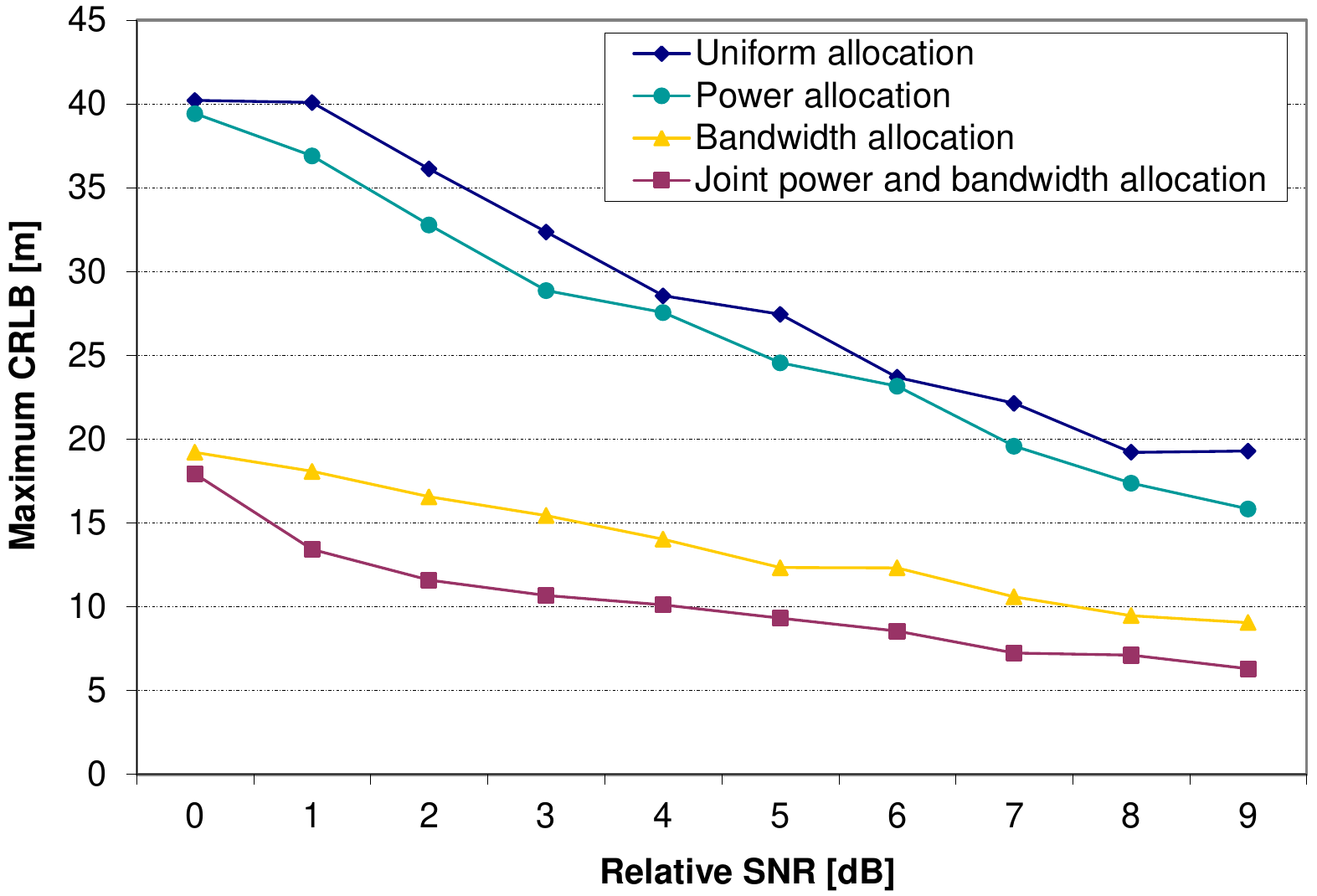}
	\caption{Square root maximum CRLB vs.\ SNR after resource allocation.}
	\label{fig:crlb}
\end{figure}

To validate the results based on the CRLB, localization errors are computed also for a multilateration algorithm implemented to estimate the target locations. Multilateration is comprised of three steps. In the first step, the time of arrivals (TOA's) of the transmitted pulses are estimated at all the receivers. Here, we apply the WRELAX algorithm \cite{Li98} to perform the task. The TOA information is then transmitted to a fusion center, where an algorithm associates TOA's to targets. Finally, using the TOA's, the target locations are estimated using the BLUE method in \cite{Godrich10-2}. For the simulation we chose for pulse shape the square root of the Hamming window in the frequency domain \cite{Oppenheim89}, which can be shown to have a good mainlobe to secondary lobe ratio necessary for TOA estimation. For a given simulation all elements are positioned randomly in the area. After positioning the elements, we run the allocation algorithms, then simulate the signals at the receivers and preform target localization by multilateration. To compare it with the cost function \eqref{eq:true_cost_function}, we save the maximum localization error among all targets and average it for all 1000 instantiations.

\begin{figure}[ptb]
	\includegraphics[width=\columnwidth]{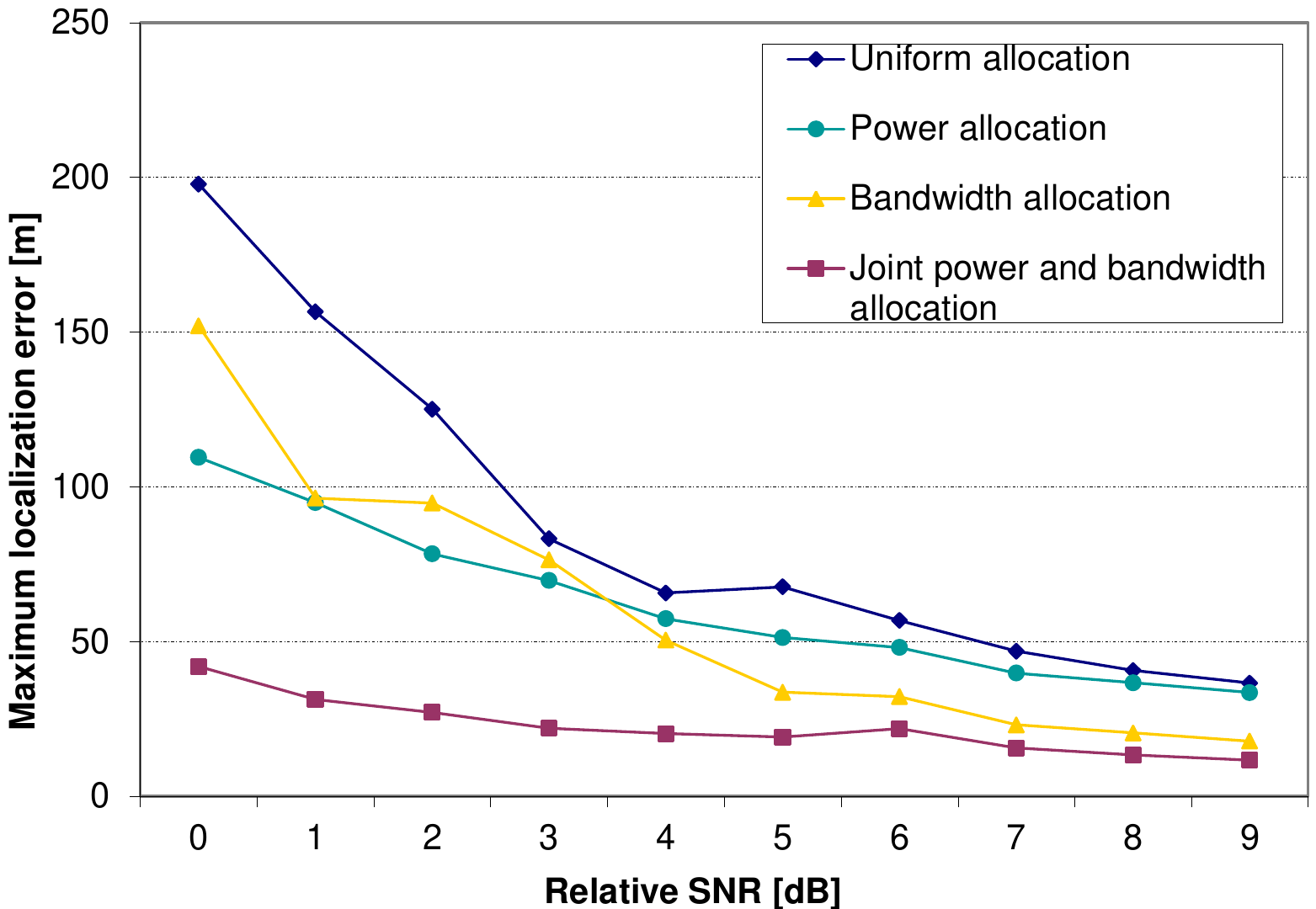}
	\caption{Square root maximum localization error vs.\ SNR after resource allocation.}
	\label{fig:mse}
\end{figure}

Fig.~\ref{fig:mse} plots the square root of the maximum localization error versus SNR. Here an increase of \SI{1}{dB} is simply used to denote an increase of \SI{1}{dB} in total power. Obviously the localization error decreases with the increase in SNR. A threshold effect is observed approximately at around \SI{4}{\decibel} of relative SNR with slight variations for the different allocation algorithms. On the right of this threshold the maximum localization error tightens over the maximum CRLB in Fig.~\ref{fig:crlb}, even though it still maintains a gap for all SNR values. The reason for this gap is, to the best knowledge of the authors, because there do not exist multilateration techniques that converge to the CRLB in the presence of multiple targets.

\subsection{Number of active transmitters}

The resource allocation algorithms distribute among the transmitters power, bandwidth, or both. Fig.~\ref{fig:active} plots the relative frequency of the number of active transmitters (transmitters whose assigned power and bandwidth is different than zero) for the three types of resource allocation. The SNR is not specified because, for all allocations the SNR simply scales the amount of resources to be assigned, but does not change which transmitters are selected.

\begin{figure}[ptb]
	\includegraphics[width=\columnwidth]{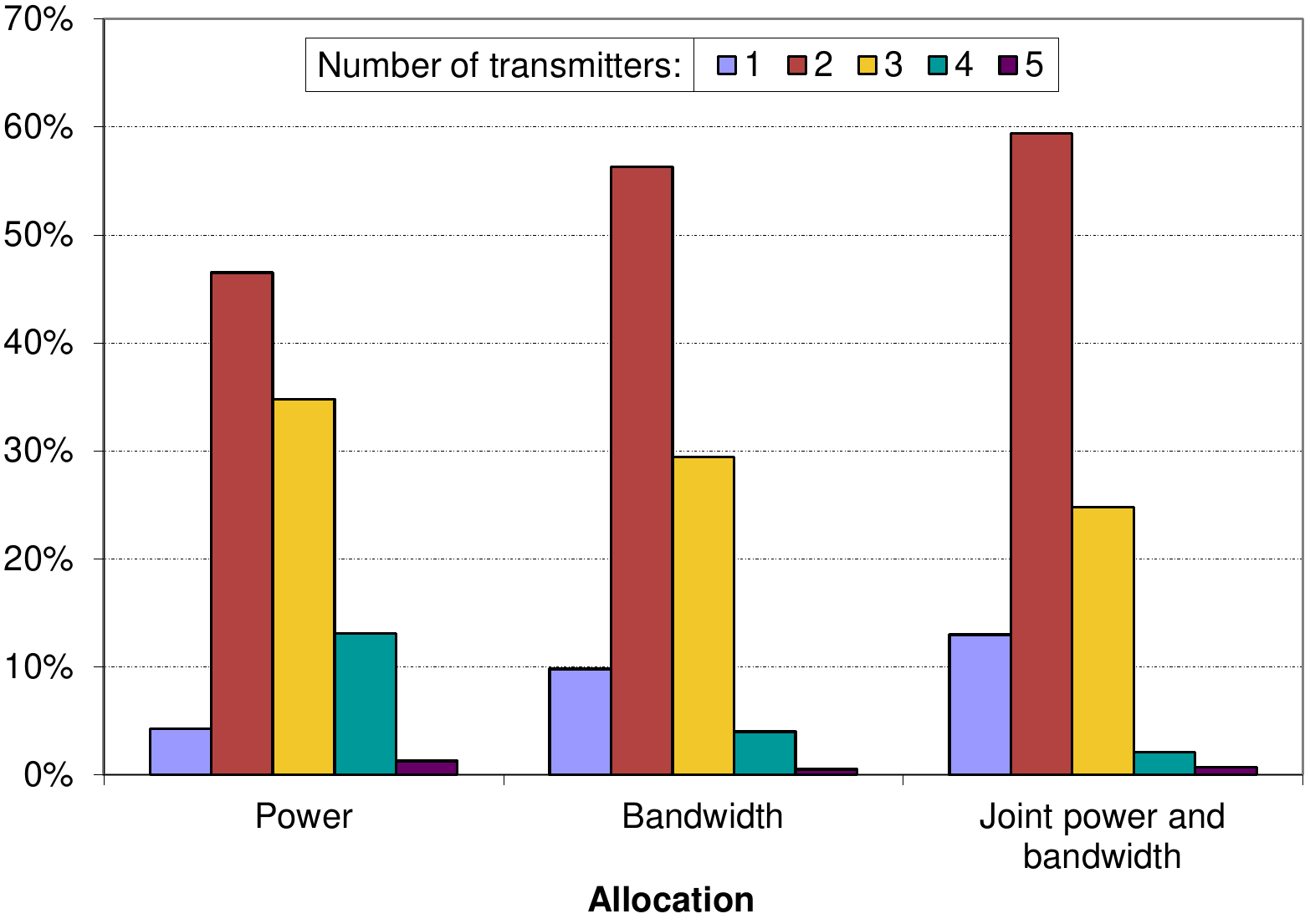}
	\caption{Relative frequencies of the number of active transmitters.}
	\label{fig:active}
\end{figure}

\subsection{Numerical evaluation of the approximations}

The purpose of this section is to validate two approximations used in this paper. The first one is due to the fact that the resource allocation algorithms employ an approximate closed-form formula for the CRLB of multiple targets \eqref{eq:diag_C} as explained in section~\ref{sec:model}. The second approximation was to consider that the solutions of our algorithms are almost as good as the optimal (but unknown) solutions. For this purpose we developed the lower-bounds in Section~\ref{sec:optimality}.

We perform simulations for all three types of resource allocations, and plot in Fig.~\ref{fig:lower-bound}, the maximum CRLB \eqref{eq:true_cost_function}, the approximate maximum CRLB \eqref{eq:cost_function}, and the lower-bound on the optimal CRLB for a prescribed total power. The figure shows how the approximate maximum CRLB and the true maximum CRLB do not vary more than 10\%. It also shows how the lower-bound is tight to the approximate maximum CRLB for the case of power allocation, thus in this case the algorithm is performing optimally. For the bandwidth and joint allocation cases, the lower-bound is not as tight, being the separation greater for the joint case, meaning another joint power-bandwidth allocation algorithm could perhaps perform slightly better.

\begin{figure}[ptb]
	\includegraphics[width=\columnwidth]{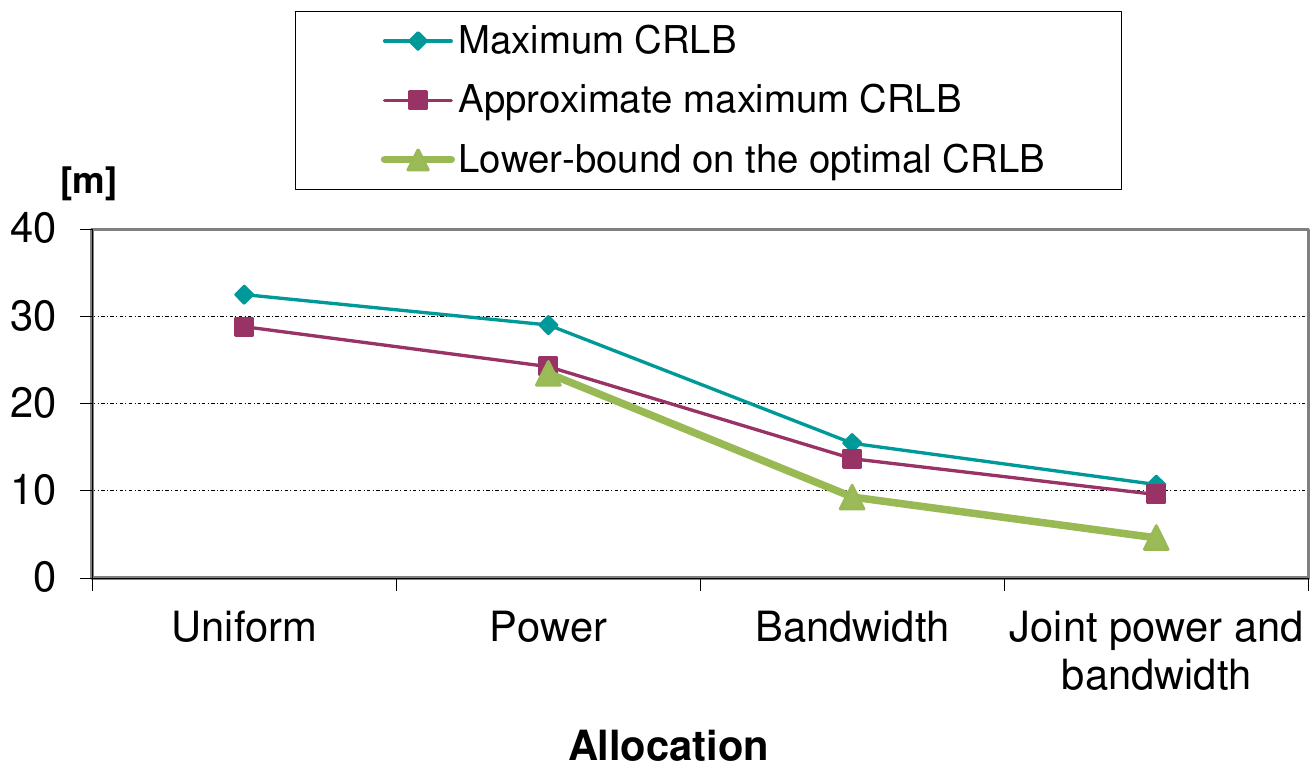}
	\caption{Comparison of the maximum CRLB, the approximate maximum CRLB, and the lower-bounds on the approximate maximum CRLB}
	\label{fig:lower-bound}
\end{figure}


\section{Conclusions}
\label{sec:conclusions}

Three approximate solutions for the allocation of power or/and bandwidth are provided given that transmitters access the medium using disjoint bandwidths of the spectrum. Extensive simulations are run on the performance of resource allocation for different SNR's in terms of the theoretical CRLB and tested by a multilateration algorithm. The best accuracy is achieved by jointly optimizing power and bandwidth allocation among the MIMO radar elements, being bandwidth a much more valuable resource than power.


\appendices

\section{}
\label{appen:proportionallity}

The power allocation $\mathbf{p}_{opt}^j$ and the bandwidth allocation $\mathbf{w}_{opt}^j$ as solutions for Problem~\ref{prob:joint} must satisfy $\mathbf{p}_{opt}^j\leq P$ and $\mathbf{w}_{opt}^j\leq B$, and must be colinear. To prove the latter we define two colinear allocations $\hat{\mathbf{p}}$ and $\hat{\mathbf{w}}$  whose $m$-th components are defined as
\begin{align} \label{eq:colinear}
	\hat{p}(m) &= \frac{P}{\sum\limits_{i=1}^{M} \big(p_{opt}^{j}(i)\big)^{\frac{1}{3}} \big(w_{opt}^{j}(i)\big)^{\frac{2}{3}}} \big(p_{opt}^{j}(m)\big)^{\frac{1}{3}} \big(w_{opt}^{j}(m)\big)^{\frac{2}{3}} \\
	\hat{w}(m) &= \frac{B}{\sum\limits_{i=1}^{M} \big(p_{opt}^{j}(i)\big)^{\frac{1}{3}} \big(w_{opt}^{j}(i)\big)^{\frac{2}{3}}} \big(p_{opt}^{j}(m)\big)^{\frac{1}{3}} \big(w_{opt}^{j}(m)\big)^{\frac{2}{3}}
\end{align}
The cost \eqref{eq:cost_function} associated to these new allocations can be written in terms of the cost function of the previous allocations using \eqref{eq:CRLB_trace}
\begin{equation} \label{eq:cost_comparison}
	g(\hat{\mathbf{p}},\hat{\mathbf{w}}) =
	\frac{\left[\sum\limits_{i=1}^{M} \big(p_{opt}^{j}(i)\big)^{\frac{1}{3}} \big(w_{opt}^{j}(i)\big)^{\frac{2}{3}}\right]^3}{P B^2} g(\mathbf{p}_{opt}^j,\mathbf{w}_{opt}^j)
\end{equation}
By H\"{o}lder's inequality \cite{Hardy34}, and using the fact that $p_{opt}^{j}(i),w_{opt}^{j}(i)\geq 0$, the numerator in the above fraction satisfies
\begin{equation}
	\left[\sum\limits_{i=1}^{M} \big(p_{opt}^{j}(i)\big)^{\frac{1}{3}} \big(w_{opt}^{j}(i)\big)^{\frac{2}{3}}\right]^3 \leq \left[(\mathbf{1}^\top\mathbf{p}_{opt}^j)^{\frac{1}{3}}(\mathbf{1}^\top\mathbf{w}_{opt}^j)^{\frac{2}{3}}\right]^3
\end{equation}
Since the allocations must satisfy $\mathbf{1}^\top\mathbf{p}_{opt}^j\leq P$  and $\mathbf{w}_{opt}^j\leq B$, it follows easily that $\left[(\mathbf{1}^\top\mathbf{p}_{opt}^j)^{\frac{1}{3}}(\mathbf{1}^\top\mathbf{w}_{opt}^j)^{\frac{2}{3}}\right]^3 \allowbreak \leq P B^2$, and therefore \eqref{eq:cost_comparison} reads
\begin{equation}
	g(\hat{\mathbf{p}},\hat{\mathbf{w}}) \leq g(\mathbf{p}_{opt}^j,\mathbf{w}_{opt}^j)
\end{equation}
Thus for any power and bandwidth allocation, it exists always a colinear alternative solution $(\hat{\mathbf{p}},\hat{\mathbf{w}})$ performing equal or better, and such that $\hat{\mathbf{w}} = \frac{B}{P}\hat{\mathbf{p}}$. Hence, the solution $(\mathbf{p}_{opt}^j,\mathbf{w}_{opt}^j)$ {must also satisfy this property, i.e. $\mathbf{w}_{opt}^j =\frac{B}{P}\mathbf{p}_{opt}^j$.

\section{}
\label{appen:equivalent}

The proof of Proposition~\ref{prop:power_dual} is by contradiction. Call $\mathbf{y}'$ the solution to problem \eqref{eq:x_reformulated}. Suppose that problem \eqref{eq:unified} admits a solution $\mathbf{y}^*$ that is better than any scaled copy of $\mathbf{y}'$ satisfying the constraints of problem \eqref{eq:unified}; i.e.\ $\mathbf{y}^*$ better than $\alpha\mathbf{y}'$ for any $\alpha$ such that $0<\alpha < \frac{D}{\mathbf{1}^\top\mathbf{y}'}$. As a better solution, the objective function evaluated at $\mathbf{y}^*$ must be smaller than evaluated at $\alpha\mathbf{y}'$: $\max_{q} g_q (\mathbf{y}^*,k) < \max_{q} g_q (\alpha\mathbf{y}',k)$ for $\alpha\in[0,\frac{D}{\mathbf{1}^\top\mathbf{y}'}]$. For the right side of this inequality, using the definition of $g(\mathbf{y},k)$ \eqref{eq:x}, we can put the constant $\alpha$  as a factor in front of the $\max$ operator: $\max_{q} g_q (\mathbf{y}^*,k) < \frac{1}{\alpha^{k+1}}\max_{q} g_q (\mathbf{y}',k)$. The most limiting value of $\alpha$ is $\alpha = \frac{D}{\mathbf{1}^\top\mathbf{y}'}$, thus leading to
\begin{equation} \label{eq:temp_ineq}
	\max_{q} g_q (\mathbf{y}^*,k) < \left(\frac{\mathbf{1}^\top\mathbf{y}'}{D}\right)^{k+1} \max_{q} g_q (\mathbf{y}',k) \leq \left(\frac{\mathbf{1}^\top\mathbf{y}'}{D}\right)^{k+1}
\end{equation}
where we use the fact that $\mathbf{y}'$ must satisfy the constraints of problem \eqref{eq:x_reformulated}, i.e.\ $\max_{q} g_q (\mathbf{y}',k)\leq 1$. The allocation policy $\mathbf{y}'' = [\max_{q} g_q (\mathbf{y}^*,k)]^{\nicefrac{1}{k+1}} \mathbf{y}^*$ based on $\mathbf{y}^*$, satisfies $g_q (\mathbf{y}'',k) = 1$. Computing the sum of $\mathbf{y}''$'s components and making use of \eqref{eq:temp_ineq} it is obtained
\begin{equation} \label{eq:temp_ineq2}
	\mathbf{1}^\top\mathbf{y}'' = \left[\max_{q} g_q (\mathbf{y}^*,k)\right]^{\nicefrac{1}{k+1}} \mathbf{1}^\top\mathbf{y}^*  < \frac{\mathbf{1}^\top\mathbf{y}'}{D} \mathbf{1}^\top\mathbf{y}^*
\end{equation}
Because $\mathbf{y}^*$ is a solution to problem \eqref{eq:unified} it satisfies $\mathbf{1}^\top\mathbf{y}^*\leq D$, and consequently in \eqref{eq:temp_ineq2},  $\mathbf{1}^\top\mathbf{y}''<  \mathbf{1}^\top\mathbf{y}'$. This indicates that $\mathbf{y}''$, rather than $\mathbf{y}'$, is a solution to problem \eqref{eq:unified} and contradicts the original hypothesis.


\begin{IEEEbiographynophoto}{Nil Garcia}
	received the M.Sc.\ degree in telecommunications engineering from the Universitat Polit\`{e}cnica de Catalunya (UPC) in Barcelona in 2008.
	
	During 2009, he worked as an engineer in the Payload department of the Centre National d'\'{e}tudes Spatiales in Toulouse (CNES), France. He is currently a Ph.D.\ candidate for the New Jersey Institute of Technology (NJIT) in Newark, and also for the National Polytechnic Institute of Toulouse (INPT). His research interests are in MIMO radar and source localization.
\end{IEEEbiographynophoto}

\begin{IEEEbiographynophoto}{Alexander M. Haimovich}
	(S'82--M'87--F'13) received the B.Sc.\ degree in electrical engineering from The Technion–Israeli Institute of Technology, Haifa, in 1977, the M.Sc.\ degree in electrical engineering from Drexel University, Philadelphia, in 1983, and the Ph.D.\ degree in systems from the University of Pennsylvania, Philadelphia, in 1989.
	
	He is the Ying Wu endowed chair professor in the Electrical and Computer Engineering at the New Jersey Institute of Technology (NJIT) in Newark. His research interests include MIMO radar and communication systems, wireless networks, and source localization.
\end{IEEEbiographynophoto}

\begin{IEEEbiographynophoto}{Martial Coulon}
	received the Ingenieur degree in computer science and applied mathematics from Ecole Nationale Supérieure d’Electronique, d'Electrotechnique, d'Informatique et d'Hydraulique et T\'{e}l\'{e}communications, and the M.Sc. degree in signal processing from the National Polytechnic Institute of Toulouse (INPT), France, both in 1996. He received the Ph.D.\ degree from INPT in 1999.
	
	During 2001--2013, he was an Assistant and Associate Professor, respectively, with the University of Toulouse, (INP--ENSEEIHT, Department of Communications Systems), and a member of the IRIT laboratory (CNRS). He is currently a Professor and dean of the communications and networks department. His research activities are centered around
	statistical signal processing and communications systems, with a particular interest to change-point detection, Bayesian inference, spread	spectrum systems, and multiuser detection.	
\end{IEEEbiographynophoto}

\begin{IEEEbiographynophoto}{Marco Lops}
	(M'96--SM'01) received the ``Laurea'' and the Ph.D.\ degrees from ``Federico II'' University, Naples.
	
	During 1989--1991 and 1991--2000, he was an Assistant and Associate Professor, respectively, with ``Federico II'' University. Since March 2000, he has been a Professor at the University of Cassino and during 2009--2010, also with ENSEEIHT, Toulouse, France. In fall 2008, he was a Visiting Professor with the University of Minnesota and in spring 2009, he was at Columbia University, New York, NY. His research interests are in detection and estimation, with emphasis on communications and radar signal processing.
\end{IEEEbiographynophoto}

\end{document}